\documentclass{article}
\usepackage[T1]{fontenc}
\usepackage[latin9]{inputenc}
\usepackage[b5paper]{geometry}
\usepackage{array}
\usepackage{amsmath}
\usepackage{amsthm}
\usepackage{amssymb}
\usepackage{stmaryrd}

\makeatletter

\providecommand{\tabularnewline}{\\}

\theoremstyle{plain}
\newtheorem{thm}{\protect\theoremname}
\theoremstyle{plain}
\newtheorem{fact}[thm]{\protect\factname}
\theoremstyle{definition}
\newtheorem{defn}[thm]{\protect\definitionname}
\theoremstyle{plain}
\newtheorem{prop}[thm]{\protect\propositionname}
\theoremstyle{plain}
\newtheorem{cor}[thm]{\protect\corollaryname}
\theoremstyle{remark}
\newtheorem{rem}[thm]{\protect\remarkname}
\theoremstyle{plain}
\newtheorem{lem}[thm]{\protect\lemmaname}

\date{}
\usepackage{latexsym}

\makeatother

\providecommand{\corollaryname}{Corollary}
\providecommand{\definitionname}{Definition}
\providecommand{\factname}{Fact}
\providecommand{\lemmaname}{Lemma}
\providecommand{\propositionname}{Proposition}
\providecommand{\remarkname}{Remark}
\providecommand{\theoremname}{Theorem}

\begin{document}

\title{Modal Logic via Global Consequence}
\author{Xuefeng Wen\\
Institute of Logic and Cognition\\
Sun Yat-sen University\\
wxflogic@gmail.com}
\maketitle
\begin{abstract}
In modal logic, semantic consequence is usually defined locally by
truth preservation at all worlds in all models (with respect to a
class of frames). It can also be defined globally by truth preservation
in all models (with respect to a class of frames). The latter is called
global consequence, which is much less studied than the standard local
one. In this paper we first study the relationship between local and
global consequence. Then we give some correspondence results for global
consequence. Finally, we illustrate two applications of global consequence,
connecting it with informational consequence and update consequence
proposed in formal semantics. Some results in the paper are already
known, which are collected in the paper for the sake of completeness.
The others appear to be new. We suggest that global consequence is
not only interesting theoretically, but also useful for application.
\end{abstract}

\section{Introduction}

Given a class of frames $\mathsf{F}$, the inference from $\Gamma$
to $\varphi$ is valid with respect to $\mathsf{F}$, if for every
world $w$ in every model $\mathfrak{M}=(W,R,V)$ such that $\mathfrak{F}=(W,R)$
is a frame in $\mathsf{F}$, if all formulas in $\Gamma$ are true
at $w$ in $\mathfrak{M}$ then $\varphi$ is also true at $w$ in
$\mathfrak{M}$. This is called the local consequence (or local validity)
in modal logic, which is the standard one. Another notion called global
consequence (or global validity) in modal logic is also defined in
the literature (e.g. in \cite{Blackburn2001a}). The inference from
$\Gamma$ to $\varphi$ is globally valid with respect to $\mathsf{F}$,
if for every model $\mathfrak{M}=(W,R,V)$ such that $\mathfrak{F}=(W,R)$
is a frame in $\mathsf{F}$, if all formulas in $\Gamma$ are true
in $\mathfrak{M}$ then $\varphi$ is also true in $\mathfrak{M}$,
where a formula is true in a model if it is true at all worlds in
the model. Compare to local consequence, global consequence is much
less studied. Notable exceptions include, \cite{Kracht2007}, \cite{Fitting1983}
and \cite{Ma2019}. Kracht \cite{Kracht2007} studied global consequence
from an algebraic point of view systematically. Fitting \cite{Fitting1983}
integrated local and global consequence into a ternary relation, and
proved completeness for various kinds of proof systems. Ma and Chen
\cite{Ma2019} presented Gentzen-style sequent calculi for global
consequence. This paper studies global consequence within the standard
relational semantics of modal logic, emphasizing its connection with
local consequence and some other consequence notions, which were proposed
for natural language arguments.

In the sequel, we consider only normal modal logics. Let $\mathcal{L}_{0}$
be the classical propositional language, $\mathcal{L}_{\Box}$ the
basic modal language. We use $\vDash$ (with or without subscripts)
for local consequence and $\vDash^{g}$ (with or without subscripts)
for global consequence, respectively. We use $\Vdash$ for satisfaction
relation. We write $\mathfrak{M},w\Vdash\Gamma$ if $\mathfrak{M},w\Vdash\varphi$
for all $\varphi\in\Gamma$. We write $\mathfrak{M}\Vdash\varphi$
if $\mathfrak{M},w\Vdash\varphi$ for all $w$ in $\mathfrak{M}$,
and $\mathfrak{F}\Vdash\varphi$ if $\mathfrak{M}\Vdash\varphi$ for
all $\mathfrak{M}$ based on $\mathfrak{F}$. We denote by $\vdash_{\mathbf{S}}$
the (local) syntactic consequence for the axiomatic system $\mathbf{S}$.
We denote by $\mathsf{K}$ be the class of all frames, and $\mathsf{M}$
the class of all models. We assume the readers are familiar with notations
for typical classes of frames and axiomatic systems. For example,
$\mathsf{K4}$ refers to the class of transitive frames, and $\mathsf{S5}$
the class of frames with equivalent relations; $\mathbf{K4}$ and
$\mathbf{S5}$ denote their corresponding axiomatic systems, respectively.
Some other notations: $\Box^{0}\varphi=\varphi$, $\Box^{n+1}\varphi=\Box\Box^{n}\varphi$,
$\Box_{r}\varphi:=\varphi\land\Box\varphi$, $\Box\Gamma:=\{\Box\varphi\mid\varphi\in\Gamma\}$,
$\Box_{r}\Gamma:=\{\Box_{r}\varphi\mid\varphi\in\Gamma\}$, $\Box^{\omega}\Gamma:=\{\Box^{n}\psi\mid n\in\mathbb{N},\psi\in\Gamma\}$,
$\Box^{\omega}\varphi:=\Box^{\omega}\{\varphi\}$.

The remaining part of the paper is organized as follows. Section~\ref{sec:global-local}
shows the relationship between local consequence and global consequence.
Section~\ref{sec:Global-Correspondence} gives a general correspondence
result for global consequence and its typical instances. Section~\ref{sec:Applications}
illustrates two applications of global consequence, connecting it
with informational consequence and update consequence proposed in
formal semantics. Section~\ref{sec:Conclusion} concludes the paper.
Some results in the paper are already known, which are collected in
the paper for the sake of completeness. The others are supposed to
be new.

\section{\label{sec:global-local}Relationship Between Local and Global Consequence}

For a start, the following are well known results that connect local
and global consequence.
\begin{fact}
For any class of frames $\mathsf{F}$, for any $\Gamma\cup\{\varphi\}\subseteq\mathcal{L}_{\Box}$, 
\begin{enumerate}
\item $\vDash_{\mathsf{F}}^{g}\varphi$ iff $\vDash_{\mathsf{F}}\varphi$;
\item $\Gamma\vDash_{\mathsf{F}}\varphi$ implies $\Gamma\vDash_{\mathsf{F}}^{g}\varphi$.
\end{enumerate}
\end{fact}

Since local and global valid formulas coincide, we are more interested
in global consequence rather than globally valid formulas. The following
fact can be easily verified, which says that local consequence and
global consequence coincide for modal-free formulas. This may be the
reason why for modal-free reasoning, we do not distinguish local and
global consequence.
\begin{fact}
\label{fact:modal-free}Let $\Gamma\cup\{\varphi\}\subseteq\mathcal{L}_{0}$.
Then for any class of frames, $\Gamma\vDash_{\mathsf{F}}^{g}\varphi$
iff $\Gamma\vDash_{\mathsf{F}}\varphi$.
\end{fact}

The following two known results show that if we add some global operators
in the language, then global consequence can always be defined by
local consequence. Before that we need two definitions for the global
operators.
\begin{defn}
Given a model $\mathfrak{M}=(W,R,V)$, define the operator $\boxplus$
as follows,
\[
\mathfrak{M},w\Vdash\boxplus\varphi\text{ iff for all }u\in R^{*}(w),\mathfrak{M},u\Vdash\varphi,
\]
where $R^{*}$ is the reflexive and transitive closure of $R$.
\end{defn}

\begin{defn}
\label{def:universal}Given a model $\mathfrak{M}=(W,R,V)$, define
the universal operator $A$ as follows,
\[
\mathfrak{M},w\Vdash A\varphi\text{ iff for all }u\in W,\mathfrak{M},u\Vdash\varphi.
\]
\end{defn}

\begin{prop}
[\cite{Venema1992}, p. 159]\label{prop:global-by-local-boxplus}For
any class of frames $\mathsf{F}$, $\Gamma\vDash_{\mathsf{F}}^{g}\varphi$
iff $\boxplus\Gamma\vDash_{\mathsf{F}}\boxplus\varphi$.
\end{prop}

\begin{prop}
[\cite{Goranko1992}, Proposition 2.1]\label{prop:global-by-local-universal}For
any class of frames $\mathsf{F}$, $\Gamma\vDash_{\mathsf{F}}^{g}\varphi$
iff $A\Gamma\vDash_{\mathsf{F}}\varphi$ iff $A\Gamma\vDash_{\mathsf{F}}A\varphi$.
\end{prop}

If we consider only the class of frames $\mathsf{K}$, then global
consequence can be defined by local consequence within the basic modal
language, as the following proposition shows.
\begin{prop}
[\cite{Blackburn2001a}, p. 32]\label{prop:global-by-local-K}For
any $\Gamma\cup\{\varphi\}\subseteq\mathcal{L}_{\Box}$, $\Gamma\vDash_{\mathsf{K}}^{g}\varphi$
iff $\Box^{\omega}\Gamma\vDash_{\mathsf{K}}\varphi$.
\end{prop}

The proposition appears as an exercise in \cite{Blackburn2001a}.
Instead of proving it directly, we generalize it as follows.
\begin{thm}
\label{thm:global-by-local}Let $\mathsf{F}$ be any class of frames
that is closed under point generated subframes. Then for any $\Gamma\cup\{\varphi\}\subseteq\mathcal{L}_{\Box}$,
$\Gamma\vDash_{\mathsf{F}}^{g}\varphi$ iff $\Box^{\omega}\Gamma\vDash_{\mathsf{F}}\varphi$.
\end{thm}

\begin{proof}
$\Rightarrow)$ Suppose $\Box^{\omega}\Gamma\nvDash_{\mathsf{F}}\varphi$.
Then there exist a frame $\mathfrak{F}$ in $\mathsf{F}$, a valuation
$V$ on $\mathfrak{F}$, and a world $w$ in $\mathfrak{F}$ such
that $\mathfrak{F},V,w\Vdash\Box^{\omega}\Gamma$ but $\mathfrak{F},V,w\nVdash\varphi$.
Let $(\mathfrak{F}',V')$ be the model generated by $w$ from $(\mathfrak{F},V)$.
Then $\mathfrak{F}',V',w\Vdash\Box^{\omega}\Gamma$ and $\mathfrak{F}',V',w\nVdash\varphi$.
From the former, it follows that $\mathfrak{F}',V'\Vdash\Gamma$,
since all worlds in $\mathfrak{F}'$ are accessible from $w$ in finite
(including zero) steps. From the latter, it follows that $\mathfrak{F}',V'\nVdash\varphi$.
Since $\mathsf{F}$ is closed under subframes, $\mathfrak{F}'$ is
also in $\mathsf{F}$. Thus, $\Gamma\nvDash_{\mathsf{F}}^{g}\varphi$.

$\Leftarrow)$ Suppose $\Gamma\nvDash_{\mathsf{F}}^{g}\varphi$. Then
there exists a frame $\mathfrak{F}$ in $\mathsf{F}$ and a valuation
$V$ on $\mathfrak{F}$ such that $\mathfrak{F},V\Vdash\Gamma$ but
$\mathfrak{F},V\nVdash\varphi$. From the latter, it follows that
there exists a world $w$ in $\mathfrak{F}$ such that $\mathfrak{F},V,w\nVdash\varphi$.
From the former, it follows that every $\psi\in\Gamma$ is true at
all worlds in $\mathfrak{F}$. Thereby, it can be easily verified
by induction that $\Box^{n}\psi$ is true at all worlds in $\mathfrak{F}$
for all $\psi\in\Gamma$ and $n\in\mathbb{N}$. In particular, $\mathfrak{F},V,w\Vdash\Box^{\omega}\Gamma$.
Hence, $\Box^{\omega}\Gamma\nvDash_{\mathsf{F}}\varphi$.
\end{proof}
Note that the direction from right to left does not require $\mathsf{F}$
to be closed under point generated subframes. The other direction,
however, does not hold for all $\mathsf{F}$, as the following fact
shows.
\begin{fact}
There exist a class of frames $\mathsf{F}$ and formulas $\Gamma\cup\{\varphi\}$
such that $\Gamma\vDash_{\mathsf{F}}^{g}\varphi$ but $\Box^{\omega}\Gamma\nvDash_{\mathsf{F}}\varphi$.
\end{fact}

\begin{proof}
Let $\mathsf{F}=\{\mathfrak{F}\}$ with $\mathfrak{F}=(\{w,u\},(w,u))$.
Then for any valuation $V$ on $\mathfrak{F}$, $\mathfrak{F},V\nVdash\Box\bot$,
since $\mathfrak{F},V,w\nVdash\Box\bot$. Hence, $\Box\bot\vDash_{\mathsf{F}}^{g}\bot$.
On the other hand, given any valuation $V$ on $\mathfrak{F}$, $\mathfrak{F},V,u\Vdash\Box^{n}\Box\bot$
for all $n\in\mathbb{N}$, but $\mathfrak{F},V,u\nVdash\bot$. Hence,
$\Box^{\omega}\Box\bot\nvDash_{\mathsf{F}}\bot$.
\end{proof}
In \cite[p. 425]{deRijke2006}, the authors claim that the equivalence
between $\Gamma\vDash_{\mathsf{F}}^{g}\varphi$ and $\Box^{\omega}\Gamma\vDash_{\mathsf{F}}\varphi$
holds for all $\mathsf{F}$, which is incorrect by the above fact.
But the closure under point generated subframes is not a necessary
condition for the equivalence in Theorem~\ref{thm:global-by-local},
as the following fact shows.
\begin{fact}
There exists a class of frames $\mathsf{F}$ that is not closed under
point generated subframes such that for any $\Gamma\cup\{\varphi\}\subseteq\mathcal{L}_{\Box}$,
$\Gamma\vDash_{\mathsf{F}}^{g}\varphi$ iff $\Box^{\omega}\Gamma\vDash_{\mathsf{F}}\varphi$.
\end{fact}

\begin{proof}
Consider $\mathsf{F}=\{\mathfrak{F}\}$ with $\mathfrak{F}=(\{w,u\},\emptyset)$.
Obviously, $\mathsf{F}$ is not closed under point generated subframes.
The direction from right to left is easy. For the other direction,
suppose $\Box^{\omega}\Gamma\nvDash_{\mathsf{F}}\varphi$. Then there
exists a valuation $V$ on $\mathfrak{F}$ such that either $\mathfrak{F},V,w\Vdash\Box^{\omega}\Gamma$
and $\mathfrak{F},V,w\nVdash\varphi$, or $\mathfrak{F},V,u\Vdash\Box^{\omega}\Gamma$
and $\mathfrak{F},V,u\nVdash\varphi$. W.l.o.g., suppose the former
holds. Then $\mathfrak{F},V,w\Vdash\Gamma$. Let $V'$ be a valuation
such that $V'(w)=V'(u)=V(w)$. It is easily verified that $\mathfrak{F},V,w\Vdash\psi$
iff $\mathfrak{F},V',w\Vdash\psi$ iff $\mathfrak{F},V',u\Vdash\psi$
for all $\psi\in\mathcal{L}_{\Box}$. Hence, $\mathfrak{F},V'\Vdash\Gamma$
and $\mathfrak{F},V'\nVdash\varphi$. Thereby, $\Gamma\nvDash_{\mathsf{F}}^{g}\varphi$.
\end{proof}
If we consider transitive frames, then the biconditional between local
consequence and global consequence can be further simplified, as the
following corollary shows.
\begin{cor}
\label{cor:transitive}Let $\mathsf{F}$ be any class of transitive
frames that is closed under point generated subframes. Then for any
$\Gamma\cup\{\varphi\}\subseteq\mathcal{L}_{\Box}$, $\Gamma\vDash_{\mathsf{F}}^{g}\varphi$
iff $\Box_{r}\Gamma\vDash_{\mathsf{F}}\varphi$.
\end{cor}

\begin{proof}
It follows from Theorem~\ref{thm:global-by-local}, noting that for
any transitive frame $\mathfrak{F}$, $\mathfrak{F}\Vdash\Box^{n}\varphi\leftrightarrow\Box\varphi$
for $n\geq1$. (Recall that $\Box_{r}\varphi$ denotes $\varphi\land\Box\varphi$.)
\end{proof}
To define global consequence by local consequence using only $\Box$
rather than $\Box_{r}$, we could add another constraint for the class
of frames.
\begin{defn}
A class of frames $\mathsf{F}$ is \emph{closed under irreflexive
point extension}, if for any frame $\mathfrak{F}=(W,R)$ in $\mathsf{F}$,
for any $w\in W$ with $\neg Rww$, any point extension $\mathfrak{F}'=(W',R')$
of $\mathfrak{F}$ for $w$ by $u\notin W$ is also in $\mathsf{F}$,
where $\mathfrak{F}'$ is defined as follows:
\[
\begin{aligned}W' & =W\cup\{u\}\\
R' & =R\cup\{(u,w)\}\cup\{(u,w')\mid(w,w')\in R\}
\end{aligned}
\]
\end{defn}

\begin{thm}
\label{thm:global-by-local-point-extension}Let $\mathsf{F}$ be any
class of transitive frames that is closed under point generated subframes
and irreflexive point extension. Then for any $\Gamma\cup\{\varphi\}\subseteq\mathcal{L}_{\Box}$,
$\Gamma\vDash_{\mathsf{F}}^{g}\varphi$ iff $\Box\Gamma\vDash_{\mathsf{F}}\Box\varphi$.
\end{thm}

\begin{proof}
$\Rightarrow)$ Suppose $\Box\Gamma\nvDash_{\mathsf{F}}\Box\varphi$.
Then there exit a frame $\mathfrak{F}$ in $\mathsf{F}$, a valuation
$V$ on $\mathfrak{F}$, and a world $w$ in $\mathfrak{F}$ such
that $\mathfrak{F},V,w\Vdash\Box\Gamma$ and $\mathfrak{F},V,w\nVdash\varphi$.
From the latter, it follows that there exists $u\in R(w)$ such that
$\mathfrak{F},V,u\nVdash\varphi$. Then from the former, it follows
that $\mathfrak{F},V,u\Vdash\Gamma\cup\Box\Gamma$, noting that $\mathfrak{F}$
is transitive. Let $\mathfrak{M}_{u}$ be the submodel of $(\mathfrak{F},V)$
generated by $u$. Then $\mathfrak{M}_{u},u\nVdash\varphi$ and hence
$\mathfrak{M}_{u},\nVdash\varphi$. Since $\mathfrak{M}_{u}$ is transitive,
every world in $\mathfrak{M}_{u}$ is either $u$ or accessible from
$u$. Thus $\mathfrak{M}_{u},v\Vdash\Gamma$ for all $v$ in $\mathfrak{M}_{u}$.
Then we have $\mathfrak{M}_{u}\Vdash\Gamma$. Since $\mathsf{F}$
is closed under point generated subframes, the frame underlying $\mathfrak{M}_{u}$
is also in $\mathsf{F}$. Therefore, $\Gamma\nvDash_{\mathsf{F}}^{g}\varphi$.

$\Leftarrow)$ Suppose $\Gamma\nvDash_{\mathsf{F}}^{g}\varphi$. Then
there exist a frame $\mathfrak{F}$ in $\mathsf{F}$ and a valuation
$V$ on $\mathfrak{F}$ such that $\mathfrak{F},V\Vdash\Gamma$ and
$\mathfrak{F},V\nVdash\varphi$. From the latter, it follows that
there exists $w$ in $\mathfrak{F}$ such that $\mathfrak{F},V,w\nVdash\varphi$.
If $Rww$, then $\mathfrak{F},V,w\nVdash\Box\varphi$. Since $\mathfrak{F},V\Vdash\Gamma$,
we also have $\mathfrak{F},V,w\Vdash\Box\Gamma$. Hence $\Box\Gamma\nvDash_{\mathsf{F}}\Box\varphi$.
If $\neg Rww$, let $\mathfrak{F}'$ be a point extension of $\mathfrak{F}$
for $w$ by $u$. Then it can be verified that $\mathfrak{F}',V,u\Vdash\Box\Gamma$
and $\mathfrak{F}',V,u\nVdash\Box\varphi$. Since $\mathsf{F}$ is
closed under irreflexive point extension, $\mathfrak{F}'$ is also
in $\mathsf{F}$. Hence, $\Box\Gamma\nvDash_{\mathsf{F}}\Box\varphi$.
\end{proof}
\begin{cor}
\label{cor:reflexive-transitive}Let $\mathsf{F}$ be any class of
reflexive and transitive frames that is closed under point generated
subframes. Then for any $\Gamma\cup\{\varphi\}\subseteq\mathcal{L}_{\Box}$,
$\Gamma\vDash_{\mathsf{F}}^{g}\varphi$ iff $\Box\Gamma\vDash_{\mathsf{F}}\Box\varphi$
iff $\Box\Gamma\vDash_{\mathsf{F}}\varphi$.
\end{cor}

\begin{proof}
Note that any class of reflexive frames is also closed under irreflexive
point extension. Thus we have the first biconditional from Theorem~\ref{thm:global-by-local-point-extension}.
The direction from left to right of the second biconditional follows
from the fact that in any reflexive frame $\mathfrak{F}$, $\mathfrak{F}\Vdash\Box\varphi\to\varphi$.
The other direction follows from the fact that for any class of frames
$\mathsf{F}$, $\Gamma\vDash_{\mathsf{F}}\varphi$ implies $\Box\Gamma\vDash_{\mathsf{F}}\Box\varphi$,
and hence $\Box\Gamma\vDash_{\mathsf{F}}\varphi$ implies $\Box\Box\Gamma\vDash_{\mathsf{F}}\Box\varphi$.
Then using the fact that for any transitive frame $\mathfrak{F}$,
$\mathfrak{F}\Vdash\Box\varphi\to\Box\Box\varphi$, we obtain the
final result.
\end{proof}
\begin{rem}
The above corollary can also be derived from Theorem~\ref{thm:global-by-local},
noting that in any reflexive frame $\mathfrak{F}$, $\mathfrak{F}\Vdash\Box_{r}\varphi\leftrightarrow\Box\varphi$.
\end{rem}

\begin{cor}
\label{cor:K4-S5}For any $\Gamma\cup\{\varphi\}\subseteq\mathcal{L}_{\Box}$,
for any $\mathsf{F}$ in $\{\mathsf{K4},\mathsf{KD4},\mathsf{S4},\mathsf{S5}\}$,
$\Gamma\vDash_{\mathsf{F}}^{g}\varphi$ iff $\Box\Gamma\vDash_{\mathsf{F}}\Box\varphi$.
\end{cor}

\begin{proof}
Straightforward from Theorem~\ref{thm:global-by-local-point-extension},
noting that all $\mathsf{F}$ in $\{\mathsf{K4},\mathsf{KD4},\mathsf{S4},\mathsf{S5}\}$
are closed under point generated subframes and irreflexive point extension.
\end{proof}
The following proposition shows that to define global consequence
by local consequence, sometimes various classes of frames are attainable.
\begin{prop}
\label{prop:s5-global}For any $\Gamma\cup\{\varphi\}\subseteq\mathcal{L}_{\Box}$,
$\Gamma\vDash_{\mathsf{S5}}^{g}\varphi$ iff $\Box\Gamma\vDash_{\mathsf{S5}}\Box\varphi$
iff $\Box\Gamma\vDash_{\mathsf{S5}}\varphi$ iff $\Box\Gamma\vDash_{\mathsf{K45}}\Box\varphi$
iff $\Box\Gamma\vDash_{\mathsf{KD45}}\Box\varphi$
\end{prop}

\begin{proof}
The first two `iff's follow from Corollary~\ref{cor:reflexive-transitive}.
The direction from right to left of the third `iff' is easy. For the
other direction, suppose $\Box\Gamma\nvDash_{\mathsf{K45}}\Box\varphi$.
Then there exist a transitive and Euclidean model $\mathfrak{M}=(W,R,V)$
and a world $w\in W$ such that $\mathfrak{M},w\Vdash\Box\Gamma$
and $\mathfrak{M},w\nVdash\Box\varphi$. From the latter, it follows
that there exists $u\in R(w)$ such that $\mathfrak{M},u\nVdash\varphi$.
Since $\mathfrak{M}$ is transitive, we also have $\mathfrak{M},u\Vdash\Box\Gamma$.
Let $\mathfrak{M}_{u}$ be the point generated submodel of $\mathfrak{M}$
by $u$. Then it can be verified that $\mathfrak{M}_{u}$ is reflexive,
transitive and Euclidean. Moreover, $\mathfrak{M}_{u},u\Vdash\Box\Gamma$
and $\mathfrak{M}_{u},u\nVdash\varphi$. Therefore $\Box\Gamma\nvDash_{\mathsf{S5}}\varphi$.
The last `iff' can be proved analogously.
\end{proof}
If we restrict premises to be modal-free formulas, then global consequence
can always be defined by local consequence (within the basic modal
language), as the following proposition shows.
\begin{prop}
\label{prop:premise-modal-free}Let $\Gamma\subseteq\mathcal{L}_{0}$
and $\varphi\in\mathcal{L}_{\Box}$. Then for any class of frames
$\mathsf{F}$, $\Gamma\vDash_{\mathsf{F}}^{g}\varphi$ iff $\Box^{\omega}\Gamma\vDash_{\mathsf{F}}\varphi$.
\end{prop}

\begin{proof}
$\Rightarrow)$ Suppose $\Box^{\omega}\Gamma\nvDash_{\mathsf{F}}\varphi$.
Then there exist a frame $\mathfrak{F}=(W,R)$ in $\mathsf{F}$, a
valuation $V$ on $\mathfrak{F}$, and a world $w$ in $\mathfrak{F}$
such that $\mathfrak{F},V,w\Vdash\Box^{\omega}\Gamma$ but $\mathfrak{F},V,w\nVdash\varphi$.
Let $(\mathfrak{F}',V')$ be the model generated by $w$ from $(\mathfrak{F},V)$.
Let $\mathfrak{F}'=(W',R')$. Then $\mathfrak{F}',V',w\Vdash\Box^{\omega}\Gamma$
and $\mathfrak{F}',V',w\nVdash\varphi$. From the former, it follows
that $\mathfrak{F}',V'\Vdash\Gamma$, since all worlds in $\mathfrak{F}'$
are accessible from $w$ in finite (including zero) steps. From the
latter, it follows that $\mathfrak{F}',V'\nVdash\varphi$. Noting
that $\Gamma$ is satisfiable and contains no modal formulas, we can
define a valuation $V''$ on $\mathfrak{F}$ such that for all worlds
in $W'$, $V''$ coincides with $V'$, and for all worlds $u$ in
$W-W'$, for every atom $p$, $u\in V''(p)$ iff $w\in V'(p)$. Then
$\mathfrak{F},V''\Vdash\Gamma$, but $\mathfrak{F},V''\nVdash\varphi$.
Thus, $\Gamma\nvDash_{\mathsf{F}}^{g}\varphi$.

$\Leftarrow)$ The same as that in the proof of Theorem~\ref{thm:global-by-local}.
\end{proof}
Some of the above results can also be given syntactically. Before
that, we need some definitions. We define local syntactic consequence
in an eliminational way, as in most textbooks in modal logic (e.g.
\cite{Chellas1980} and \cite{Blackburn2001a}) , i.e. $\Gamma\vdash_{\mathbf{S}}\varphi$
iff there is a finite subset $\Delta\subseteq\Gamma$ such that $\vdash_{\mathbf{S}}\bigwedge\Delta\to\varphi$.
The gist of this definition is to prevent the application of the rule
of necessitation to the premises in $\Gamma$, since the inference
from $\varphi$ to $\Box\varphi$ is generally not valid under local
semantic consequence. On the contrary, since we have $\varphi\vDash^{g}\Box\varphi$,
given a standard axiomatic system, the global syntactic consequence
$\vdash_{\mathbf{S}}^{g}$ can be defined in the same way as in classical
propositional logic, i.e. $\Gamma\vdash_{\mathbf{S}}^{g}\varphi$
iff there is finite sequence of formulas $\varphi_{1},\ldots,\varphi_{n}$
such that $\varphi_{n}=\varphi$ and for each $i\le n$ either $\varphi_{i}\in\Gamma$,
or $\varphi_{i}$ is an instance of an axiom scheme, or $\varphi_{i}$
is obtained from previous formulas in the sequence by applying the
rule(s) of the system. As a result, under global syntactic consequence,
the rule of necessitation is applicable to the premises. Now we have
the following result.
\begin{prop}
Let $\mathbf{S}$ be any axiomatic extension of $\mathbf{K}$. Then
for any $\Gamma\cup\{\varphi\}\subseteq\mathcal{L}_{\Box}$, $\Gamma\vdash_{\mathbf{S}}^{g}\varphi$
iff $\Box^{\omega}\Gamma\vdash_{\mathbf{S}}\varphi$.
\end{prop}

\begin{cor}
Let $\mathbf{S}$ be any axiomatic extension of $\mathbf{K}4$. Then
for any $\Gamma\cup\{\varphi\}\subseteq\mathcal{L}_{\Box}$, $\Gamma\vdash_{\mathbf{S}}^{g}\varphi$
iff $\Box_{r}\Gamma\vdash_{\mathbf{S}}\varphi$.
\end{cor}

\begin{cor}
Let $\mathbf{S}$ be any axiomatic extension of $\mathbf{S4}$. Then
for any $\Gamma\cup\{\varphi\}\subseteq\mathcal{L}_{\Box}$, $\Gamma\vdash_{\mathbf{S}}^{g}\varphi$
iff $\Box\Gamma\vdash_{\mathbf{S}}\varphi$ iff $\Box\Gamma\vdash_{\mathbf{S}}\Box\varphi$.
\end{cor}

These results can be obtained by the completeness of the axiomatic
systems as well as the above semantic results. They can also be proved
directly by induction on the length of proofs. We omit it here.

Conversely, local consequence can also be defined by global consequence,
but much harder. We need a local operator.
\begin{defn}
Given a model $\mathfrak{M}$, define the `only' operator as follows:
\[
\mathfrak{M},w\Vdash O\varphi\text{ iff }\mathfrak{M},w\Vdash\varphi\text{ and for all }w'\neq w,\mathfrak{M},w'\nVdash\varphi.
\]
\end{defn}

Venema gave the following result in \cite{Venema1992} (without proof).
\begin{prop}
[\cite{Venema1992}, p. 159]\label{prop:local-by-global}For any
class of frames $\mathsf{F}$, for any $\Gamma\cup\{\varphi\}\subseteq\mathcal{L}_{\Box EO}$,
for any $p\notin Var(\Gamma\cup\{\varphi\})$
\[
\Gamma\vDash_{\mathsf{F}}\varphi\text{ iff }\{EOp\}\cup\{p\to\gamma\mid\gamma\in\Gamma\}\vDash_{\mathsf{F}}^{g}p\to\varphi,
\]
where $E$ is the dual of the universal operator $A$ in Definition~\ref{def:universal}.
\end{prop}

We summarize the results in this section as follows. Those with bold
fonts are supposed to be new.
\begin{center}
\begin{tabular}{l|l>{\raggedright}p{0.4\textwidth}}
 & Local by Global & Global by Local\tabularnewline
\hline 
restricting $\mathcal{L}_{\Box}$, for all $\mathsf{F}$ & Fact.~\ref{fact:modal-free} & Fact.~\ref{fact:modal-free}, \textbf{Prop}.~\textbf{\ref{prop:premise-modal-free}}\tabularnewline
beyond $\mathcal{L}_{\Box}$, for all $\mathsf{F}$ & Prop.~\ref{prop:local-by-global} & Prop.~\ref{prop:global-by-local-boxplus}, Prop.~\ref{prop:global-by-local-universal}\tabularnewline
within $\mathcal{L}_{\Box}$, for some $\mathsf{F}$ &  & \textbf{Thm}.~\textbf{\ref{thm:global-by-local}}, \textbf{Thm}.~\textbf{\ref{thm:global-by-local-point-extension}},
Prop.~\ref{prop:global-by-local-K}, Cor.~\ref{cor:transitive},
Cor.~\ref{cor:reflexive-transitive}, Cor.~\ref{cor:K4-S5}, \textbf{Prop}.~\textbf{\ref{prop:s5-global}}\tabularnewline
\end{tabular}
\par\end{center}

Though within $\mathcal{L}_{\Box}$ global consequence can not be
reduced to local consequence generally, many properties for local
consequence are preserved for global consequence. See \cite{Kracht1999,Kracht2007,Kracht2011}.

\section{\label{sec:Global-Correspondence}Global Correspondence}

If we consider the correspondence between modal formulas and first-order
frame properties, then there is nothing new for global consequence,
since globally valid formulas coincide with locally valid formulas.
But if consider the correspondence between modally valid inferences
and first-order frame properties, then it turns out to be much different
for global consequence.

First, we have the following obvious fact.
\begin{fact}
\label{fact:phi2box}$\varphi\vDash_{\mathsf{F}}^{g}\Box\varphi$
for any class of frames $\mathsf{F}$, in particular, we have
\begin{enumerate}
\item $\Box\varphi\vDash_{\mathsf{F}}^{g}\Box\Box\varphi$
\item $\Diamond\varphi\vDash_{\mathsf{F}}^{g}\Box\Diamond\varphi$
\end{enumerate}
\end{fact}

In contrast, $\Box\varphi\vDash_{\mathsf{F}}\Box\Box\varphi$ if and
only if $\mathsf{F}$ is transitive, and $\Diamond\varphi\vDash_{\mathsf{F}}\Box\Diamond\varphi$
if and only if $\mathsf{F}$ is Euclidean.
\begin{fact}
\label{fact:globally-isolated}$\Diamond\varphi\vDash_{\mathsf{F}}^{g}\varphi$
iff $\mathsf{F}$ is globally isolated, i.e. for every $\mathfrak{F}=(W,R)$
in $\mathsf{F}$, $\forall x\exists y\forall z(Ryz\to z=x)$.
\end{fact}

\begin{proof}
$\Leftarrow)$ Given any globally isolated frame $\mathfrak{F}=(W,R)$
in $\mathsf{F}$, given any valuation $V$ on $\mathfrak{F}$, suppose
$\mathfrak{F},V\Vdash\Diamond\varphi$. Given any $x\in W$, since
$\mathfrak{F}$ is globally isolated, there exists $y\in W$ s.t.
for all $z\in W$, if $Ryz$ then $z=x$. Since $\mathfrak{F},V\Vdash\varphi$,
we have $\mathfrak{F},V,y\Vdash\Diamond\varphi$. Then there exists
$z\in W$ s.t. $Ryz$ and $\mathfrak{F},V,z\Vdash\varphi$. By the
property of $R$, we have $z=x$. Hence, $\mathfrak{F},V,x\Vdash\varphi$.
Since $x$ is arbitrary, we have $\mathfrak{F},V\Vdash\varphi$, as
required.

$\Rightarrow)$ Suppose $\mathfrak{F}=(W,R)$ in $\mathsf{F}$ is
not globally isolated. Then there exists $x\in W$ s.t. for all $y\in W$
there exists $z\in W$ s.t. $Ryz$ and $z\neq x$. Let $V(p)=W-\{x\}$.
Then $\mathfrak{F},V,x\nVdash p$ and hence $\mathfrak{F},V\nVdash p$.
Given any $y\in W$, by the property of $R$, there exists $z\neq x$
s.t. $Ryz$. Hence, $\mathfrak{F},V,z\Vdash p$ and thus $\mathfrak{F},V,y\Vdash\Diamond p$.
Since $y$ is arbitrary, we have $\mathfrak{F},V\Vdash\Diamond p$.
Therefore, $\Diamond p\nvDash_{\mathsf{F}}^{g}p$.
\end{proof}
\begin{fact}
$\Diamond\Diamond\varphi\vDash_{\mathsf{F}}^{g}\Diamond\varphi$ iff
$\mathsf{F}$ is globally transitive, i.e. for every $\mathfrak{F}=(W,R)$
in $\mathsf{F}$, $\forall w\exists x\forall y\forall z(Rxy\land Ryz\to Rwz)$.
\end{fact}

\begin{proof}
$\Leftarrow)$ Given any globally transitive frame $\mathfrak{F}=(W,R)$
in $\mathsf{F}$, given any valuation $V$ on $\mathfrak{F}$, suppose
$\mathfrak{F},V\Vdash\Diamond\Diamond\varphi$. Given any $w\in W$,
since $\mathfrak{F}$ is globally transitive, there exists $x\in W$
s.t. for any $y,z\in W$ if $Rxy$ and $Ryz$ then $Rwz$. By $\mathfrak{F},V\Vdash\Diamond\Diamond\varphi$,
we have $\mathfrak{F},V,x\Vdash\Diamond\Diamond\varphi$. Then there
exists $y,z\in W$ s.t. $Rxy$, $Ryz$, and $\mathfrak{F},V,z\Vdash\varphi$.
By the property of $R$, $Rwz$. Hence, $\mathfrak{F},V,w\Vdash\Diamond\varphi$.
Since $w$ is arbitrary, $\mathfrak{F},V\Vdash\Diamond\varphi$, as
required.

$\Rightarrow)$ Suppose $\mathfrak{F}=(W,R)$ in $\mathsf{F}$ is
not globally transitive. Then there exists $w\in W$ s.t. for all
$x\in W$, there exist $y,z\in W$ s.t. $Rxy$, $Ryz$, and $\neg Rwz$.
Let $V(p)=W-R(x)$. Then $\mathfrak{F},V,w\nVdash\Diamond p$ and
hence $\mathfrak{F},V\nVdash\Diamond p$. Given any $x\in W$, since
there exist $y,z\in W$ s.t. $Rxy$, $Ryz$, and $\neg Rwz$, we have
$\mathfrak{F},V,z\Vdash p$ and hence $\mathfrak{F},V,x\Vdash\Diamond\Diamond p$.
Since $x$ is arbitrary, we have $\mathfrak{F},V\Vdash\Diamond\Diamond p$.
Therefore, $\Diamond\Diamond p\nvDash_{\mathsf{F}}^{g}\Diamond p$.
\end{proof}
\begin{fact}
$\Diamond\Box\varphi\vDash_{\mathsf{F}}^{g}\Box\varphi$ iff $\mathsf{F}$
is globally Euclidean, i.e. for every $\mathfrak{F}=(W,R)$ in $\mathsf{F}$,
$\forall w\forall x\exists y\forall z(Rwx\land Ryz\to Rzx)$.
\end{fact}

\begin{proof}
$\Leftarrow)$ Given any globally Euclidean frame $\mathfrak{F}=(W,R)$
in $\mathsf{F}$, given any valuation $V$ on $\mathfrak{F}$, suppose
$\mathfrak{F},V\Vdash\Diamond\Box\varphi$. Given any $w\in W$, suppose
$Rwx$. Since $\mathfrak{F}$ is globally Euclidean, there exists
$y\in W$ s.t. for all $z\in W$ if $Rwx$ and $Ryz$ then $Rzx$.
By $\mathfrak{F},V\Vdash\Diamond\Box\varphi$, we have $\mathfrak{F},V,y\Vdash\Diamond\Box\varphi$.
Then there exists $z\in W$ s.t. $\mathfrak{F},V,z\Vdash\Box\varphi$.
By the property of $R$, we have $Rzx$. It follows that $\mathfrak{F},V,x\Vdash\varphi$.
Thus, $\mathfrak{F},V,w\Vdash\Box\varphi$. Since $w$ is arbitrary,
we have $\mathfrak{F},V\Vdash\Box\varphi$, as required.

$\Rightarrow)$ Suppose $\mathfrak{F}=(W,R)$ in $\mathsf{F}$ is
not globally Euclidean. Then there exists $w,x\in W$ s.t. $Rwx$
and for all $y\in W$ there exists $z\in W$ s.t. $Ryz$, and $\neg Rzx$.
Let $V(p)=W-\{x\}$. Then $\mathfrak{F},V,w\nVdash\Box p$ and hence
$\mathfrak{F},V\nVdash\Box p$. Given any $y\in W$, by the property
of $R$, there exists $z\in W$ s.t. $Ryz$, and $\neg Rzx$. Hence,
$\mathfrak{F},V,z\Vdash\Box p$ and $\mathfrak{F},V,y\Vdash\Diamond\Box p$.
Since $y$ is arbitrary, we have $\mathfrak{F},V\Vdash\Diamond\Box p$.
Therefore, $\Diamond\Box p\nvDash_{\mathsf{F}}^{g}\Box p$.
\end{proof}
\begin{fact}
$\Box\varphi\vDash_{\mathsf{F}}^{g}\varphi$ iff $\mathsf{F}$ is
globally reflexive, i.e. for every $\mathfrak{F}=(W,R)$ in $\mathsf{F}$,
$\forall x\exists yRyx$.
\end{fact}

\begin{proof}
$\Leftarrow)$ Given any globally reflexive frame $\mathfrak{F}=(W,R)$
in $\mathsf{F}$, given any valuation $V$ on $\mathfrak{F}$, suppose
$\mathfrak{F},V\Vdash\Box\varphi$. Given any $x\in W$, since $\mathfrak{F}$
is backward serial, there exists $y\in W$ s.t. $Ryx$. Since $\mathfrak{F},V\Vdash\Box\varphi$,
we have $\mathfrak{F},V,y\Vdash\Box\varphi$. Hence, $\mathfrak{F},V,x\Vdash\varphi$.
Since $x$ is arbitrary, we have $\mathfrak{F},V\Vdash\varphi$, as
required.

$\Rightarrow)$ Suppose $\mathfrak{F}=(W,R)$ in $\mathsf{F}$ is
not globally reflexive. Then exists $x\in W$ s.t. for all $y\in W$,
$x\notin R(y)$. Let $V(p)=W-\{x\}$. Then $\mathfrak{F},V,x\nVdash p$
and hence $\mathfrak{F},V\nVdash p$. Given any $y\in W$, since $x\notin R(y)$,
we have $\mathfrak{F},V,y\Vdash\Box p$. Since $y$ is arbitrary,
we have $\mathfrak{F},V\Vdash\Box p$. Therefore, $\Box p\nvDash_{\mathsf{F}}^{g}p$.
\end{proof}
\begin{fact}
$\varphi\vDash_{\mathsf{F}}^{g}\Diamond\varphi$ iff $\mathsf{F}$
is globally inverse reflexive, i.e. for every $\mathfrak{F}=(W,R)$
in $\mathsf{F}$, $\forall x\exists yRxy$.
\end{fact}

\begin{proof}
$\Leftarrow)$ Given any globally inverse reflexive frame $\mathfrak{F}=(W,R)$
in $\mathsf{F}$, given any valuation $V$ on $\mathfrak{F}$, suppose
$\mathfrak{F},V\Vdash\varphi$. Then $\mathfrak{F},V,w\Vdash\varphi$
for all $w\in W$. Since $\mathfrak{F}$ is globally inverse reflexive,
$\mathfrak{F},V,w\Vdash\Diamond\varphi$ for all $w\in W$, i.e. $\mathfrak{F},V\Vdash\Diamond\varphi$.

$\Rightarrow)$ Suppose $\mathfrak{F}=(W,R)$ in $\mathsf{F}$ is
not globally inverse reflexive. Then there exists $w\in W$ such that
$R(w)=\emptyset$. Let $V(p)=W$. Then $\mathfrak{F},V\Vdash p$ but
$\mathfrak{F},V,w\nVdash\Diamond p$. Thus $\mathfrak{F},V\nVdash\Diamond p$.
Therefore, $p\nvDash_{\mathsf{F}}^{g}\Diamond p$.
\end{proof}
\begin{fact}
$\Box\varphi\vDash_{\mathsf{F}}^{g}\Diamond\varphi$ iff $\mathsf{F}$
is globally serial, i.e. for every $\mathfrak{F}=(W,R)$ in $\mathsf{F}$,
$\forall x\exists y\exists z(Ryz\land Rxz)$.
\end{fact}

\begin{proof}
$\Leftarrow)$ Given any globally serial frame $\mathfrak{F}=(W,R)$
in $\mathsf{F}$, given any valuation $V$ on $\mathfrak{F}$, suppose
$\mathfrak{F},V\Vdash\Box\varphi$. Given any $x\in W$, since $\mathfrak{F}$
is globally serial, there exist $y,z\in W$ s.t. $Ryz$ and $Rxz$.
By $\mathfrak{F},V\Vdash\Box\varphi$, we have $\mathfrak{F},V,y\Vdash\Box\varphi$.
Hence, $\mathfrak{F},V,z\Vdash\varphi$. By $Rxz$, we have $\mathfrak{F},V,x\Vdash\Diamond\varphi$.
Since $x$ is arbitrary, we have $\mathfrak{F},V\Vdash\Diamond\varphi$,
as required.

$\Rightarrow)$ Suppose $\mathfrak{F}=(W,R)$ in $\mathsf{F}$ is
not globally serial. Then there exists $x\in W$ s.t. for all $y,z\in W$if
$Ryz$ then $\neg Rxz$. Let $V(p)=W-R(x)$. Then $\mathfrak{F},V,x\nVdash\Diamond p$
and hence $\mathfrak{F},V\nVdash\Diamond p$. Given any $y\in W$,
suppose $Ryz$, by the property of $R$, we have $z\notin R(x)$.
Hence, $\mathfrak{F},V,z\Vdash p$. Thus $\mathfrak{F},V,y\Vdash\Box p$.
Since $y$ is arbitrary, we have $\mathfrak{F},V\nVdash\Box p$. Therefore,
$\Box p\nvDash_{\mathsf{F}}^{g}\Diamond p$.
\end{proof}
\begin{fact}
$\varphi\vDash_{\mathsf{F}}^{g}\Box\Diamond\varphi$ iff $\mathsf{F}$
is globally symmetric, i.e. for every $\mathfrak{F}=(W,R)$ in $\mathsf{F}$,
$\forall x\forall y\exists z(Rxy\to Ryz)$.
\end{fact}

\begin{proof}
$\Leftarrow)$ Given any globally symmetric frame $\mathfrak{F}=(W,R)$
in $\mathsf{F}$, given any valuation $V$ on $\mathfrak{F}$, suppose
$\mathfrak{F},V\Vdash\varphi$. Given any $x\in W$, suppose $Rxy$.
Since $\mathfrak{F}$ is globally symmetric, there exists $z\in W$
s.t. $Ryz$. Since $\mathfrak{F},V\Vdash\varphi$, we have $\mathfrak{F},V,z\Vdash\varphi$.
Thus $\mathfrak{F},V,y\Vdash\Diamond\varphi$. Hence, $\mathfrak{F},V,x\Vdash\Box\Diamond\varphi$.
Since $x$ is arbitrary, we have $\mathfrak{F},V\Vdash\Box\Diamond\varphi$,
as required.

$\Rightarrow)$ Suppose $\mathfrak{F}=(W,R)$ in $\mathsf{F}$ is
not globally symmetric. Then there exists $x,y\in W$ s.t. $Rxy$
and $R(y)=\emptyset$. Let $V(p)=W$. Then $\mathfrak{F},V\Vdash p$
and $\mathfrak{F},V,y\nVdash\Diamond p$. Thus $\mathfrak{F},V,x\nVdash\Box\Diamond p$
and hence $\mathfrak{F},V\nVdash\Box\Diamond p$. Therefore, $p\nvDash_{\mathsf{F}}^{g}\Box\Diamond p$.
\end{proof}
\begin{fact}
$\Diamond\Box\varphi\vDash_{\mathsf{F}}^{g}\varphi$ iff $\mathsf{F}$
is globally inverse symmetric, i.e. for every $\mathfrak{F}=(W,R)$
in $\mathsf{F}$, $\forall x\exists y\forall z(Ryz\to Rzx)$.
\end{fact}

\begin{proof}
$\Leftarrow)$ Given any globally inverse symmetric frame $\mathfrak{F}=(W,R)$
in $\mathsf{F}$, given any valuation $V$ on $\mathfrak{F}$, suppose
$\mathfrak{F},V\Vdash\Diamond\Box\varphi$. Given any $x\in W$, by
the property of $R$, there exists $y\in W$ s.t. for all $z\in W$
if $Ryz$ then $Rzx$. Since $\mathfrak{F},V\Vdash\Diamond\Box\varphi$,
we have $\mathfrak{F},V,y\Vdash\Diamond\Box\varphi$. Then there exists
$z\in W$ s.t. $\mathfrak{F},V,z\Vdash\Box\varphi$. By the property
of $R$, we have $Rzx$. Hence, $\mathfrak{F},V,x\Vdash\varphi$.
Since $x$ is arbitrary, we have $\mathfrak{F},V\Vdash\varphi$, as
required.

$\Rightarrow)$ Suppose $\mathfrak{F}=(W,R)$ in $\mathsf{F}$ is
not globally inverse symmetric. Then there exists $x\in W$ s.t. for
all $y\in W$ there exists $z\in W$ s.t. $Ryz$ and $\neg Rzx$.
Let $V(p)=W-\{x\}$. Then $\mathfrak{F},V,x\nVdash p$ and hence $\mathfrak{F},V\nVdash p$.
Given any $y\in W$, by the property of $R$, there exists $z\in W$
s.t. $Ryz$ and $\neg Rzx$. Hence, $\mathfrak{F},V,z\Vdash\Box p$
and $\mathfrak{F},V,y\Vdash\Diamond\Box p$. Since $y$ is arbitrary,
we have $\mathfrak{F},V\Vdash\Diamond\Box p$. Therefore, $\Diamond\Box p\nvDash_{\mathsf{F}}^{g}p$.
\end{proof}
Note that for local consequence, a valid inference often has an equivalent
dual version. For example, $\Box\varphi\vDash\varphi$ iff $\varphi\vDash\Diamond\varphi$.
This equivalence, however, does not hold for global consequence. For
example, though $\Box\varphi\vDash^{g}\Box\Box\varphi$ holds for
any class of frames, its dual $\Diamond\Diamond\varphi\vDash^{g}\Diamond\varphi$
holds only for globally transitive frames. This is a notable contrast
between local and global consequence.

Parallel to a famous general correspondence result for local consequence,
we give a general correspondence result for global consequence, of
which the above facts are all instances.
\begin{thm}
$\Diamond^{i}\Box^{j}\varphi\vDash_{\mathsf{F}}^{g}\Box^{k}\Diamond^{l}\varphi$
iff every frame $\mathfrak{F}=(W,R)$ in $\mathsf{F}$ satisfies the
following condition
\[
\forall w\forall x\exists y\forall z\exists u(R^{k}wx\land R^{i}yz\to R^{l}xu\land R^{j}zu).
\]
\end{thm}

\begin{proof}
$\Leftarrow)$ Given any frame $\mathfrak{F}=(W,R)$ in $\mathsf{F}$
that satisfies the above property, given any valuation $V$ on $\mathfrak{F}$,
suppose $\mathfrak{F},V\Vdash\Diamond^{i}\Box^{j}\varphi$. Given
any $w\in W$, suppose $R^{k}wx$. Then by the property of $R$, there
exists $y\in W$ s.t. for all $z\in W$ if $R^{i}yz$ then there exists
$u\in W$ s.t. $R^{l}xu$ and $R^{j}zu$. By $\mathfrak{F},V\Vdash\Diamond^{i}\Box^{j}\varphi$,
we have $\mathfrak{F},V,y\Vdash\Diamond^{i}\Box^{j}\varphi$. Then
it follows that there exists $z\in W$ s.t. $R^{i}yz$ and $\mathfrak{F},V,z\Vdash\Box^{j}\varphi$.
By the property of $R$, there exists $u\in W$ s.t. $R^{l}xu$ and
$R^{j}zu$. Thus $\mathfrak{F},V,u\Vdash\varphi$ and $\mathfrak{F},V,x\Vdash\Diamond^{l}\varphi$.
Hence, $\mathfrak{F},V,w\Vdash\Box^{k}\Diamond^{l}\varphi$. Since
$w$ is arbitrary, we have $\mathfrak{F},V\Vdash\Box^{k}\Diamond^{l}\varphi$,
as required.

$\Rightarrow)$ Suppose $\mathfrak{F}=(W,R)$ in $\mathsf{F}$ does
not satisfy the above property. Then there exists $w,x\in W$ s.t.
$R^{k}wx$ and for all $y\in W$ there exits $z\in W$ s.t. $R^{i}yz$
and $R^{l}(x)\cap R^{j}(z)=\emptyset$. Let $V(p)=W-R^{l}(x)$. Then
$\mathfrak{F},V,x\Vdash\Box^{l}\neg p$ and $\mathfrak{F},V,w\Vdash\Diamond^{k}\Box^{l}\neg p$.
Hence, $\mathfrak{F},V,w\nVdash\Box^{k}\Diamond^{l}p$ and $\mathfrak{F},V\nVdash\Box^{k}\Diamond^{l}p$.
Given any $y\in W$, by the property of $R$, there exists $z\in W$
s.t. $R^{i}yz$ and $R^{l}(x)\cap R^{j}(z)=\emptyset$. Thus $\mathfrak{F},V,z\Vdash\Box^{j}p$
and $\mathfrak{F},V,y\Vdash\Diamond^{i}\Box^{j}p$. Since $y$ is
arbitrary, we have $\mathfrak{F},V\Vdash\Diamond^{i}\Box^{j}p$. Therefore,
$\Diamond^{i}\Box^{j}p\nvDash_{\mathsf{F}}^{g}\Box^{k}\Diamond^{l}p$.
\end{proof}

\section{\label{sec:Applications}Applications}

\subsection{Informational Consequence}

In \cite{Yalcin2007} Yalcin advocated a non-classical consequence
relation, called informational consequence. Yalcin noticed that if
$\Diamond$ denotes epistemic `might' or `may', then saying both $\varphi$
and $\Diamond\neg\varphi$ seems inconsistent, which is not reflected
in standard modal logic. So he proposed domain semantics and informational
consequence (details below) to formalize this phenomenon. We will
soon find that informational consequence is intimately related to
global consequence.
\begin{defn}
A \emph{domain model} is a pair $\mathfrak{D}=(W,V)$, where $W\ne\emptyset$
and $V:PV\to\wp(W)$ is a valuation on $W$. Given a domain model
$\mathfrak{D}=(W,V)$, that $\varphi$ is true at $(w,i)\in W\times\wp(W)$
in $\mathfrak{D}$, denoted $\mathfrak{D},w,i\Vdash\varphi$, is inductively
defined as follows, where $\mathfrak{D},i\Vdash\varphi$ means for
all $w\in i$, $\mathfrak{D},w,i\Vdash\varphi$:
\begin{itemize}
\item $\mathfrak{D},w,i\Vdash p$ iff $w\in V(p)$
\item $\mathfrak{D},w,i\Vdash\neg\varphi$ iff $\mathfrak{D},w,i\nVdash\varphi$
\item $\mathfrak{D},w,i\Vdash\varphi\land\psi$ iff $\mathfrak{D},w,i\Vdash\varphi$
and $\mathfrak{D},w,i\Vdash\psi$
\item $\mathfrak{D},w,i\Vdash\Box\varphi$ iff $\mathfrak{D},i\Vdash\varphi$
\end{itemize}
\end{defn}

\begin{defn}
[Informational consequence]The inference from $\Gamma$ to $\varphi$
is \emph{informationally valid}, denoted $\Gamma\vDash_{I}\varphi$,
if for all domain models $\mathfrak{D}=(W,V)$ and $i\subseteq W$,
$\mathfrak{D},i\Vdash\Gamma$ implies $\mathfrak{D},i\Vdash\varphi$.
\end{defn}

It can be easily shown that under domain semantics, $\varphi\land\Diamond\neg\varphi\vDash_{I}\bot$.
But this can also be achieved by global consequence for free.
\begin{fact}
\label{fact:epistemic-contradiction}$\varphi\land\Diamond\neg\varphi\vDash_{\mathsf{F}}^{g}\bot$
for any class of frames $\mathsf{F}$.
\end{fact}

\begin{proof}
Suppose $\mathfrak{F},V\Vdash\varphi\land\Diamond\neg\varphi$. Then
$\mathfrak{F},V\Vdash\varphi$ and $\mathfrak{F},V\Vdash\Diamond\neg\varphi$.
The former implies that $\mathfrak{F},V\Vdash\Box\varphi$, which
contradicts the latter.
\end{proof}
In \cite{Bledin2014}, Bledin convincingly argued that the rule of
reduction to absurdity and constructive dilemma are not generally
valid for natural language arguments. Rather, their correct forms
should add some modal operators. More precisely, Bledin suggests that
\begin{itemize}
\item $\Gamma,\varphi\vDash\bot\not\Rightarrow\Gamma\vDash\neg\varphi$,
instead we have $\Gamma,\varphi\vDash\bot\Rightarrow\Gamma\vDash\Diamond\neg\varphi$;
\item $\Gamma,\alpha\vDash\varphi,\Gamma,\beta\vDash\psi\not\Rightarrow\Gamma,\alpha\lor\beta\vDash\varphi\lor\psi$,
instead we have $\Gamma,\alpha\vDash\varphi,\Gamma,\beta\vDash\psi\Rightarrow\Gamma,\Box\alpha\lor\Box\beta\vDash\Box\varphi\lor\Box\psi$.
\end{itemize}
Bledin argued that informational consequence can perfectly predict
the above desiderata. But global consequence can do the same job as
well.
\begin{fact}
\label{fact:reduction-to-absurdity}$\Gamma,\varphi\vDash_{\mathsf{F}}^{g}\bot\not\Rightarrow\Gamma\vDash_{\mathsf{F}}^{g}\neg\varphi$,
instead for any reflexive and transitive $\mathsf{F}$, we have $\Gamma,\varphi\vDash_{\mathsf{F}}^{g}\bot\Rightarrow\Gamma\vDash_{\mathsf{F}}^{g}\Diamond\neg\varphi$.
\end{fact}

\begin{proof}
By Fact~\ref{fact:epistemic-contradiction}, we have $\Diamond\neg\varphi,\varphi\vDash_{\mathsf{F}}^{g}\bot$
for any class of frames $\mathsf{F}$. But by Fact~\ref{fact:globally-isolated},
$\Diamond\neg\varphi\vDash_{\mathsf{F}}^{g}\neg\varphi$ holds only
for $\mathsf{F}$ that is globally isolated. For the remaining part,
suppose $\Gamma\nvDash_{\mathsf{F}}^{g}\Diamond\neg\varphi$. Then
there exists a model $\mathfrak{M}$ with its underlying frame in
$\mathsf{F}$ such that $\mathfrak{M}\Vdash\Gamma$ and $\mathfrak{M}\nVdash\Diamond\neg\varphi$.
By the latter there exists $w$ in $\mathfrak{M}$ such that $\mathfrak{M},w\nVdash\Diamond\neg\varphi$,
i.e. $\mathfrak{M},w\Vdash\Box\varphi$. Let $\mathfrak{M}_{w}$ be
the subframe of $\mathfrak{M}$ generated by $w$. Then $\mathfrak{M}_{w},w\Vdash\Box\varphi$.
Since $\mathfrak{M}_{w}$ is reflexive and transitive, we have $\mathfrak{M}_{w}\Vdash\varphi$.
Thus $\Gamma,\varphi\nvDash_{\mathsf{F}}^{g}\bot$.
\end{proof}
\begin{fact}
\label{fact:constructive-dilemma}$\Gamma,\alpha\vDash_{\mathsf{F}}^{g}\varphi,\Gamma,\beta\vDash_{\mathsf{F}}^{g}\psi\not\Rightarrow\Gamma,\alpha\lor\beta\vDash_{\mathsf{F}}^{g}\varphi\lor\psi$,
instead for any reflexive and transitive $\mathsf{F}$, we have $\Gamma,\alpha\vDash_{\mathsf{F}}^{g}\varphi,\Gamma,\beta\vDash_{\mathsf{F}}^{g}\psi\Rightarrow\Gamma,\Box\alpha\lor\Box\beta\vDash_{\mathsf{F}}^{g}\Box\varphi\lor\Box\psi$.
\end{fact}

\begin{proof}
By Fact~\ref{fact:phi2box}, we have $p\vDash_{\mathsf{F}}^{g}\Box p$
and $\neg p\vDash_{\mathsf{F}}^{g}\Box\neg p$ for any class of frames
$\mathsf{F}$. But it is easily verified that $p\lor\neg p\nvDash_{\{\mathfrak{F}\}}^{g}\Box p\lor\Box\neg p$,
where $\mathfrak{F}=(\{1,2\},\{(1,2),(2,1)\})$. For the remaining
part, suppose $\Gamma,\alpha\vDash_{\mathsf{F}}^{g}\varphi$ and $\Gamma,\beta\vDash_{\mathsf{F}}^{g}\psi$.
Let $\mathfrak{M}$ be any model with its underlying frame in $\mathsf{F}$.
Suppose $\mathfrak{M}\Vdash\Gamma$ and $\mathfrak{M}\Vdash\Box\alpha\lor\Box\beta$.
Given any $w$ in $\mathfrak{M}$, we have $\mathfrak{M},w\Vdash\Box\alpha\lor\Box\beta$.
Then either $\mathfrak{M},w\Vdash\Box\alpha$ or $\mathfrak{M},w\Vdash\Box\beta$.
Since $\mathfrak{M}$ is reflexive and transitive, if the former holds,
then $\mathfrak{M}_{w}\Vdash\alpha$. By $\Gamma,\alpha\vDash_{\mathsf{F}}^{g}\varphi$,
we have $\mathfrak{M}_{w}\Vdash\varphi$. Thus $\mathfrak{M},w\Vdash\Box\varphi$.
If the latter holds, then $\mathfrak{M}_{w}\Vdash\beta$. By $\Gamma,\beta\vDash_{\mathsf{F}}^{g}\psi$,
we have $\mathfrak{M}_{w}\Vdash\psi$. Thus $\mathfrak{M},w\Vdash\Box\psi$.
Hence, $\mathfrak{M},w\Vdash\Box\varphi\lor\Box\psi$. Since $w$
is arbitrary, we have $\mathfrak{M}\Vdash\Box\varphi\lor\Box\psi$,
as required.
\end{proof}
Indeed, Schulz proved the following general result.
\begin{thm}
[\cite{Schulz2010}, Theorem 2.1]\label{thm:schulz}For any $\Gamma\cup\{\varphi\}\subseteq\mathcal{L}_{\Box}$,
$\Gamma\vDash_{I}\varphi$ iff $\Box\Gamma\vDash_{\mathsf{S5}}\Box\varphi$.
\end{thm}

By Proposition~\ref{prop:s5-global}, the following corollary easily
follows.
\begin{cor}
\label{cor:s5-i}For any $\Gamma\cup\{\varphi\}\subseteq\mathcal{L}_{\Box}$,
$\Gamma\vDash_{I}\varphi$ iff $\Gamma\vDash_{\mathsf{S5}}^{g}\varphi$.
\end{cor}

Compared to Theorem~\ref{thm:schulz}, it appears that Corollary~\ref{cor:s5-i}
better characterizes informational consequence, since the former uses
local consequence and by Proposition~\ref{prop:s5-global}, with
local consequence not only $\mathsf{S5}$ can be used, but also $\mathsf{K45}$
and $\mathsf{KD45}$ are attainable. But with global consequence,
such multiple correspondence disappears. On the other hand, Facts~\ref{fact:reduction-to-absurdity}
and \ref{fact:constructive-dilemma} show that if we just need to
satisfy the desiderata above proposed by Yalcin and Bledin, it is
possible to consider only $\mathsf{S4}$ instead of $\mathsf{S5}$,
as far as global consequence is used.

\subsection{Update Consequence}

Update semantics proposed by Veltman in \cite{Veltman1996} is also
a popular semantics for natural languages. In update semantics, two
conjunctions can be defined. One is static (as in \cite{Veltman1996}),
the other dynamic (as in \cite{Willer2013,Willer2015}). To differentiate
them, we consider the following language.

Given the set of propositional variables $PV$, the language $\mathcal{L}_{\Box;}$
is defined as follows:
\[
\mathcal{L}_{\Box;}\ni\varphi::=p\mid\neg\varphi\mid(\varphi\land\varphi)\mid(\varphi;\varphi)\mid\Box\varphi,
\]
where $p\in PV$, $\land$ is the static conjunction and $;$ the
dynamic one. We stipulate that both $\land$ and $;$ are left associated,
so that $\varphi_{1}\land\varphi_{2}\land\varphi_{3}$ abbreviates
$(\varphi_{1}\land\varphi_{2})\land\varphi_{3}$, and $\varphi_{1};\varphi_{2};\varphi_{3}$
abbreviates $(\varphi_{1};\varphi_{2});\varphi_{3}$, etc.
\begin{defn}
An \emph{update model} is a pair $\mathfrak{U}=(W,V)$, where $W\ne\emptyset$
and $V:PV\to\wp(W)$ is a valuation on $W$. Given an update model
$\mathfrak{U}=(W,V)$, define the update function $\cdot[\cdot]_{\mathfrak{U}}:\wp(W)\times\mathcal{L}_{\Box;}\to\wp(W)$
on $\mathfrak{U}$ as follows.
\begin{itemize}
\item $s[p]_{\mathfrak{U}}=s\cap V(p)$
\item $s[\neg\varphi]_{\mathfrak{U}}=s-s[\varphi]_{\mathfrak{U}}$
\item $s[\varphi\land\psi]_{\mathfrak{U}}=s[\varphi]_{\mathfrak{U}}\cap s[\psi]_{\mathfrak{U}}$
\item $s[\varphi;\psi]_{\mathfrak{U}}=s[\varphi]_{\mathfrak{U}}[\psi]_{\mathfrak{U}}$
\item $s[\Box\varphi]_{\mathfrak{U}}=\{w\in s\mid s[\varphi]_{\mathfrak{U}}=s\}$
\end{itemize}
We say that $s$ supports $\varphi$ in $\mathfrak{U}$, denoted $\mathfrak{U},s\Vdash_{U}\varphi$,
if $s[\varphi]_{\mathfrak{U}}=s$. We write $\mathfrak{U},s\Vdash_{U}\Gamma$
iff $\mathfrak{U},s\Vdash_{U}\varphi$ for all $\varphi\in\Gamma$.
\end{defn}

It is easily seen that for any $\mathfrak{U}$ and $s$ in $\mathfrak{U}$,
for any $\varphi\in\mathcal{L}_{\Box;}$, $s[\varphi]_{\mathfrak{U}}\subseteq s$.
\begin{defn}
[Update consequence]We say that $\varphi$ is an \emph{update consequence}
of $\Gamma$, denoted $\Gamma\vDash_{U}\varphi$, if for all update
models $\mathfrak{U}=(W,V)$, for all information states $s\subseteq W$,
$\mathfrak{U},s\Vdash_{U}\Gamma$ implies $\mathfrak{U},s\Vdash_{U}\varphi$.
We say that $\varphi$ is a \emph{sequential update consequence} of
the sequence $\gamma_{1},\ldots,\gamma_{n}$, denoted $\gamma_{1},\ldots,\gamma_{n}\vDash_{SU}\varphi$,
if for all update models $\mathfrak{U}=(W,V)$, for all information
states $s\subseteq W$, $s[\gamma_{1}]_{\mathfrak{U}}\cdots[\gamma_{n}]_{\mathfrak{U}}\Vdash_{U}\varphi$.
\end{defn}

Sometimes, another operator $\rightarrowtriangle$ for indicative
conditionals is also defined in update semantics (e.g. \cite{Gillies2004}),
whose update function is given below.
\begin{itemize}
\item $s[\varphi\rightarrowtriangle\psi]_{\mathfrak{U}}=\{w\in s\mid s[\varphi]_{\mathfrak{U}}[\psi]_{\mathfrak{U}}=s[\varphi]_{\mathfrak{U}}\}$
\end{itemize}
It follows that $\rightarrowtriangle$ can be defined by $\Box$ and
$;$ as the following fact shows.
\begin{fact}
For all $\mathfrak{U}$ and $s$ in $\mathfrak{U}$, $s[\varphi\rightarrowtriangle\psi]_{\mathfrak{U}}=s[\Box\neg(\varphi;\neg\psi)]_{\mathfrak{U}}$.
\end{fact}

Now with $\rightarrowtriangle$, sequential update consequence can
be reduced to update consequence.
\begin{lem}
\label{lem:update-validity}For any $\gamma_{1},\ldots,\gamma_{n},\varphi\in\mathcal{L}_{\Box;}$,
$\gamma_{1},\ldots,\gamma_{n}\vDash_{SU}\varphi$ iff $\vDash_{U}(\gamma_{1};\ldots;\gamma_{n})\rightarrowtriangle\varphi$
iff $\vDash_{U}\Box\neg(\gamma_{1};\ldots;\gamma_{n};\neg\varphi)$.
\end{lem}

\begin{proof}
Straightforward from the definitions.
\end{proof}
Now we prove that update consequence can be defined by global consequence.
\begin{defn}
Given a relational model $\mathfrak{M}=(W,R,V)$, define the truth
condition for $\varphi;\psi$ as follows.
\begin{itemize}
\item $\mathfrak{M},w\Vdash\varphi;\psi$ iff $\mathfrak{M},w\Vdash\varphi$
and $\mathfrak{M}^{\varphi},w\Vdash\psi$, where $\mathfrak{M}^{\varphi}=(W^{\varphi},R^{\varphi},V^{\varphi})$
is given below:
\[
\begin{aligned}W^{\varphi} & =\{w\in W\mid\mathfrak{M},w\Vdash\varphi\}\\
R^{\varphi} & =R\cap(W^{\varphi}\times W^{\varphi})\\
V^{\varphi}(p) & =W^{\varphi}\cap V(p),\text{for all }p\in PV.
\end{aligned}
\]
\end{itemize}
\end{defn}

Given a relational model $\mathfrak{M}=(W,R,V)$, we write $\llbracket\varphi\rrbracket^{\mathfrak{M}}$
for the truth set of $\varphi$ in $\mathfrak{M}$, i.e. $\llbracket\varphi\rrbracket^{\mathfrak{M}}=\{w\in W\mid\mathfrak{M},w\Vdash\varphi\}$.
\begin{lem}
\label{lem:update-model-invariance}For any update models $\mathfrak{U}=(W,V)$
and $\mathfrak{U}'=(W',V')$ such that $W\subseteq W'$ and $V=V'\upharpoonright_{W}$,
for any $s\subseteq W$,
\[
s[\varphi]_{\mathfrak{U}}=s[\varphi]_{\mathfrak{U}'}.
\]
\end{lem}

\begin{proof}
By induction on $\varphi$.
\begin{itemize}
\item $\varphi=p$. Then $s[\varphi]_{\mathfrak{U}}=s[p]_{\mathfrak{U}}=s\cap V(p)=s\cap V'(p)=s[p]_{\mathfrak{U}'}=s[\varphi]_{\mathfrak{U}'}$.
\item The Boolean cases are easily verified.
\item $\varphi=\psi;\chi$. Then $s[\varphi]_{\mathfrak{U}}=s[\psi;\chi]_{\mathfrak{U}}=s[\psi]_{\mathfrak{U}}[\chi]_{\mathfrak{U}}=s[\psi]_{\mathfrak{U}}[\chi]_{\mathfrak{U}'}=s[\psi]_{\mathfrak{U}'}[\chi]_{\mathfrak{U}'}=s[\psi;\chi]_{\mathfrak{U}'}=s[\varphi]_{\mathfrak{U}'}$.
\item $\varphi=\Box\psi$. Then $s[\varphi]_{\mathfrak{U}}=s[\Box\psi]_{\mathfrak{U}}=\{w\in s\mid s[\varphi]_{\mathfrak{U}}=s\}=\{w\in s\mid s[\varphi]_{\mathfrak{U}'}=s\}=s[\Box\psi]_{\mathfrak{U}'}=s[\varphi]_{\mathfrak{U}'}$.
\end{itemize}
\end{proof}
\begin{lem}
\label{lem:relational2update}For any relational model $\mathfrak{M}=(W,R,V)$
with $R=W\times W$ and its underlying update model $\mathfrak{U}^{\mathfrak{M}}=(W,V)$,
for any $\varphi\in\mathcal{L}_{\Box}$,
\[
W[\varphi]_{\mathfrak{U}^{\mathfrak{M}}}=\llbracket\varphi\rrbracket^{\mathfrak{M}}.
\]
Hence, $\mathfrak{U}^{\mathfrak{M}},W\Vdash_{U}\varphi$ iff $\mathfrak{M}\Vdash\varphi$.
\end{lem}

\begin{proof}
By induction on $\varphi$.
\begin{itemize}
\item $\varphi=p\in PV$. Then $W[\varphi]_{\mathfrak{U}^{\mathfrak{M}}}=W[p]_{\mathfrak{U}^{\mathfrak{M}}}=W\cap V(p)=V(p)=\llbracket\varphi\rrbracket^{\mathfrak{M}}$.
\item The Boolean cases are easily verified.
\item $\varphi=\psi;\chi$. Then $W[\varphi]_{\mathfrak{U}^{\mathfrak{M}}}=W[\psi]_{\mathfrak{U}^{\mathfrak{M}}}[\chi]_{\mathfrak{U}^{\mathfrak{M}}}=\llbracket\psi\rrbracket^{\mathfrak{M}}[\chi]_{\mathfrak{U}^{\mathfrak{M}}}=W'[\chi]_{\mathfrak{U}^{\mathfrak{M}}}=W'[\chi]_{\mathfrak{U}'}=\llbracket\chi\rrbracket^{\mathfrak{M}^{\psi}}=\llbracket\psi;\chi\rrbracket^{\mathfrak{M}}=\llbracket\varphi\rrbracket^{\mathfrak{M}},$
where $W'=\llbracket\psi\rrbracket^{\mathfrak{M}}$ and $\mathfrak{U}'=(W',V\upharpoonright_{W'})$.
Note that the fourth identity follows from Lemma~\ref{lem:update-model-invariance}.
\item $\varphi=\Box\psi$. Then $W[\varphi]_{\mathfrak{U}^{\mathfrak{M}}}=W[\Box\psi]_{\mathfrak{U}^{\mathfrak{M}}}=\begin{cases}
W & \text{if }W[\psi]_{\mathfrak{U}^{\mathfrak{M}}}=W\\
\emptyset & \text{otherwise}
\end{cases}$

$=\begin{cases}
W & \text{if }\llbracket\psi\rrbracket^{\mathfrak{M}}=W\\
\emptyset & \text{otherwise}
\end{cases}=\begin{cases}
\llbracket\Box\psi\rrbracket^{\mathfrak{M}} & \text{if }\llbracket\psi\rrbracket^{\mathfrak{M}}=W\\
\llbracket\Box\psi\rrbracket^{\mathfrak{M}} & \text{otherwise}
\end{cases}=\llbracket\Box\psi\rrbracket^{\mathfrak{M}}=\llbracket\varphi\rrbracket^{\mathfrak{M}}$.

\end{itemize}
\end{proof}
\begin{lem}
\label{lem:update2relational}Given an update model $\mathfrak{U}=(W,V)$
and an information state $s\subseteq W$, define $\mathfrak{M}^{s}=(s,s\times s,V^{s})$,
where $V^{s}(p)=s\cap V(p)$. Then for any $\varphi\in\mathcal{L}_{\Box}$,
\[
s[\varphi]_{\mathfrak{U}}=\llbracket\varphi\rrbracket^{\mathfrak{M}^{s}}.
\]
Hence, $\mathfrak{U},s\Vdash_{U}\varphi$ iff $\mathfrak{M}^{s}\Vdash\varphi$.
\end{lem}

\begin{proof}
By induction on $\varphi$.
\begin{itemize}
\item $\varphi=p\in PV$. Then $s[\varphi]_{\mathfrak{U}}=s[p]_{\mathfrak{U}}=s\cap V(p)=V^{s}(p)=\llbracket p\rrbracket^{\mathfrak{M}^{s}}$.
\item The Boolean cases are easily verified.
\item $\varphi=\psi;\chi$. Then x$s[\varphi]_{\mathfrak{U}}=s[\psi;\chi]_{\mathfrak{U}}=s[\psi]_{\mathfrak{U}}[\chi]_{\mathfrak{U}}=\llbracket\psi\rrbracket^{\mathfrak{M}^{s}}[\chi]_{\mathfrak{U}}=s'[\chi]_{\mathfrak{U}}=\llbracket\chi\rrbracket^{\mathfrak{M}^{s'}}=\llbracket\chi\rrbracket^{(\mathfrak{M}^{s})^{\varphi}}=\llbracket\psi;\chi\rrbracket^{\mathfrak{M}^{s}}$,
where $s'=\llbracket\psi\rrbracket^{\mathfrak{M}^{s}}$.
\item $\varphi=\Box\psi$. Then $s[\varphi]_{\mathfrak{U}}=s[\Box\psi]_{\mathfrak{U}}=\begin{cases}
s & \text{if }s[\psi]_{\mathfrak{U}}=s\\
\emptyset & \text{otherwise}
\end{cases}=\begin{cases}
s & \text{if }\llbracket\psi\rrbracket^{\mathfrak{M}^{s}}=s\\
\emptyset & \text{otherwise}
\end{cases}$

$=\begin{cases}
\llbracket\Box\psi\rrbracket^{\mathfrak{M}^{s}} & \text{if }\llbracket\psi\rrbracket^{\mathfrak{M}^{s}}=s\\
\llbracket\Box\psi\rrbracket^{\mathfrak{M}^{s}} & \text{otherwise}
\end{cases}=\llbracket\Box\psi\rrbracket^{\mathfrak{M}^{s}}=\llbracket\varphi\rrbracket^{\mathfrak{M}^{s}}$.

\end{itemize}
\end{proof}
\begin{thm}
\label{thm:update-2-S5}For any $\Gamma\cup\{\varphi\}\subseteq\mathcal{L}_{\Box;}$,
$\Gamma\vDash_{U}\varphi$ iff $\Gamma\vDash_{\mathsf{S5}}^{g}\varphi$.
\end{thm}

\begin{proof}
$\Rightarrow)$ Suppose $\Gamma\nvDash_{\mathsf{S5}}^{g}\varphi$.
Then there exists an $\mathsf{S5}$ model $\mathfrak{M}$ such that
$\mathfrak{M}\Vdash\Gamma$ and $\mathfrak{M}\nVdash\varphi$. By
the latter, there exists $w$ in $\mathfrak{M}$ such that $\mathfrak{M},w\nVdash\varphi$.
Let $\mathfrak{M}_{w}=(W_{w},R_{w},V_{w})$ be the submodel of $\mathfrak{M}$
generated by $w$. Then $\mathfrak{M}_{w}\Vdash\Gamma$ and $\mathfrak{M}_{w}\nVdash\varphi$.
Since $\mathfrak{M}_{w}$ is a universal model, by Lemma~\ref{lem:relational2update},
we have $\mathfrak{U}^{\mathfrak{M}_{w}},W_{w}\Vdash_{U}\Gamma$ and
$\mathfrak{U}^{\mathfrak{M}_{w}},W_{w}\nVdash_{U}\varphi$. Hence,
$\Gamma\nvDash_{U}\varphi$.

$\Leftarrow)$ Suppose $\Gamma\nvDash_{U}\varphi$. Then there exist
an update model $\mathfrak{U}$ and an information state $s$ in $\mathfrak{U}$
such that $\mathfrak{U},s\Vdash_{U}\Gamma$ and $\mathfrak{U},s\nVdash_{U}\varphi$.
By Lemma~\ref{lem:update2relational}, we have $\mathfrak{M}^{s}\Vdash\Gamma$
and $\mathfrak{M}^{s}\nVdash\varphi$. Since $\mathfrak{M}^{s}$ is
an $\mathsf{S5}$ model, it follows that $\Gamma\nvDash_{\mathsf{S5}}^{g}\varphi$.
\end{proof}
\begin{cor}
\label{cor:sequential-update-2-S5}For any $\gamma_{1},\ldots,\gamma_{n},\varphi\in\mathcal{L}_{\Box;}$,
\[
\gamma_{1},\ldots,\gamma_{n}\vDash_{SU}\varphi\text{ iff }\vDash_{\mathsf{S5}}^{g}\Box\neg(\gamma_{1};\ldots;\gamma_{n};\neg\varphi)\text{ iff }\vDash_{\mathsf{S5}}\Box\neg(\gamma_{1};\ldots;\gamma_{n};\neg\varphi).
\]
\end{cor}

\begin{proof}
Straightforward from Lemma~\ref{lem:update-validity} and Theorem~\ref{thm:update-2-S5}.
\end{proof}
Note that the truth condition for $\varphi;\psi$ is just the same
as that for $\langle\varphi\rangle\psi$ in public announcement logic
(PAL, henceforth. For an excellent overview of PAL and more generally
dynamic epistemic logic, see \cite{Ditmarsch2007}.). Thus $\varphi\rightarrowtriangle\psi$
is just $\Box[\varphi]\psi$ in PAL. Hence, we can define the following
translation from $\mathcal{L}_{\Box;}$ to $\mathcal{L}_{PAL}$.
\begin{defn}
Define $t:\mathcal{L}_{\Box;}\to\mathcal{L}_{PAL}$ as follows.
\begin{itemize}
\item $t(p)=p$, $p\in PV$
\item $t(\neg\varphi)=\neg t(\varphi)$
\item $t(\varphi\land\psi)=t(\varphi)\land t(\psi)$
\item $t(\varphi;\psi)=\langle t(\varphi)\rangle t(\psi)$
\item $t(\Box\varphi)=\Box t(\varphi)$
\end{itemize}
\end{defn}

Now we can define $\vDash_{SU}$ by the standard local or global consequence
within $\mathcal{L}_{PAL}$.
\begin{thm}
For any $\Gamma\cup\{\gamma_{1},\ldots,\gamma_{n},\varphi\}\subseteq\mathcal{L}_{\Box;}$,
\begin{enumerate}
\item $\Gamma\vDash_{U}\varphi$ iff $t(\Gamma)\vdash_{\mathbf{PAL}}^{g}t(\varphi)$,
\item $\gamma_{1},\ldots,\gamma_{n}\vDash_{SU}\varphi$ iff $\vdash_{\mathbf{PAL}}[\langle\cdots\langle t(\gamma_{1})\rangle t(\gamma_{2})\rangle t(\gamma_{3})\cdots\rangle t(\gamma_{n})]t(\varphi).$
\end{enumerate}
\end{thm}

\begin{proof}
For (1), by Theorem~\ref{thm:update-2-S5}, we have $\Gamma\vDash_{U}\varphi$
iff $\Gamma\vDash_{\mathsf{S5}}^{g}\varphi$. Since $\varphi;\psi$
has the same truth condition as $\langle\varphi\rangle\psi$ in PAL,
we have $\Gamma\vDash_{\mathsf{S5}}^{g}\varphi$ iff $t(\Gamma)\vDash_{\mathsf{S5}}^{g}t(\varphi)$.
Then by the completeness of $\mathbf{PAL}$ (for global consequence),
we have $t(\Gamma)\vDash_{\mathsf{S5}}^{g}t(\varphi)$ iff $t(\Gamma)\vdash_{\mathbf{PAL}}^{g}t(\varphi)$.
Item (2) follows from Corollary~\ref{cor:sequential-update-2-S5}
in the same way, noting that $[\varphi]\psi\leftrightarrow\neg\langle\varphi\rangle\neg\psi$
and $\vdash_{\mathbf{PAL}}\varphi$ iff $\vdash_{\mathbf{PAL}}\Box\varphi$.
\end{proof}
It is well known that (single agent) $\mathbf{PAL}$ can be reduced
to $\mathbf{S5}$. It follows that sequential update consequence in
$\mathcal{L}_{\Box;}$ can finally be defined by the local or global
consequence of $\mathbf{S5}$ within $\mathcal{L}_{\Box}$. This in
turn implies that $\mathcal{L}_{\Box}$ has the same expressive power
as $\mathcal{L}_{\Box;}$, for both update consequence and sequential
update consequence. 
\begin{cor}
For any $\Gamma\cup\{\gamma_{1},\ldots,\gamma_{n},\varphi\}\subseteq\mathcal{L}_{\Box}$,
\begin{enumerate}
\item $\Gamma\vDash_{U}\varphi$ iff $\Gamma\vdash_{\mathbf{S5}}^{g}\varphi$,
\item $\gamma_{1},\ldots,\gamma_{n}\vDash_{SU}\varphi$ iff $\vdash_{\mathbf{PAL}}[\langle\cdots\langle\gamma_{1}\rangle\gamma_{2}\rangle\gamma_{3}\cdots\rangle\gamma_{n}]\varphi.$
\end{enumerate}
\end{cor}

\section{\label{sec:Conclusion}Concluding Remarks}

Though global consequence can be defined by local consequence for
some classes of frames, it has its independent value for application.
If domain semantics and update semantics are considered to be good
formalizations of natural languages, then global consequence could
also be useful for this application. Moreover, it is more flexible
than the former two, since we can consider different classes of frames,
which is absent in the former two semantics.

This paper is only a first step to the study of global consequence
in modal logic. For instance, is there a sufficient and necessary
condition on frames for global consequence to be defined by local
consequence? Is there a Sahlqvist-like correspondence between global
consequence and first-order properties? How to compare the frame definability
between local consequence extended with universal modality and global
consequence? We leave these technical issues for future research. 

\paragraph*{Acknowledgments.}

Thanks to audiences at the 2019 logic seminar and the mathematical
philosophy week \& workshop on modal logic at Peking University, where
I gave talks on part of this paper and received valuable feedback.
Thanks also to two anonymous referees of the 2019 national conference
on modern logic for their careful reading of my draft and their helpful
comments.


\end{document}